\documentclass{amsart}
\usepackage[utf8]{inputenc}
\usepackage{amssymb,amsmath,amsfonts}
\usepackage{amsaddr}
\usepackage{geometry}
\usepackage{multicol,multirow}
\usepackage{enumerate,enumitem,setspace}
\usepackage{verbatim}
\usepackage[lowtilde]{url}
\usepackage{xcolor}
\usepackage{graphicx}
\usepackage{anyfontsize}
\DeclareSymbolFont{rsfscript}{OMS}{rsfs}{m}{n}
\DeclareSymbolFontAlphabet{\mathrsfs}{rsfscript}
\DeclareMathOperator{\rt}{rt}
\DeclareMathOperator{\lspan}{span}
\newtheorem{theorem}{Theorem}
\newtheorem{corollary}[theorem]{Corollary}
\newtheorem{lemma}[theorem]{Lemma}
\newtheorem{proposition}[theorem]{Proposition}

\newtheorem{openproblem}{Open Problem}
\theoremstyle{remark}

\begin{document}
\title[Improving the upper bound]{Improving the upper bound on\\the length of the shortest reset words}
\author{Marek Szyku{\l}a}
\email{msz@cs.uni.wroc.pl}
\address{Institute of Computer Science,\\University of Wroc{\l}aw, Wroc{\l}aw, Poland}

\begin{abstract}
We improve the best known upper bound on the length of the shortest reset words of synchronizing automata.
The new bound is slightly better than $114 n^3 / 685 + O(n^2)$.
The \v{C}ern\'{y} conjecture states that $(n-1)^2$ is an upper bound.
So far, the best general upper bound was $(n^3-n)/6-1$ obtained by J.-E.~Pin and P.~Frankl in 1982.
Despite a number of efforts, it remained unchanged for about 35 years.

To obtain the new upper bound we utilize avoiding words.
A word is avoiding for a state $q$ if after reading the word the automaton cannot be in $q$.
We obtain upper bounds on the length of the shortest avoiding words, and using the approach of Trahtman from 2011 combined with the well-known Frankl theorem from 1982, we improve the general upper bound on the length of the shortest reset words.
For all the bounds, there exist polynomial algorithms finding a word of length not exceeding the bound. 

\end{abstract}
\maketitle
\section{Introduction}

We deal with deterministic finite complete (semi)automata $\mathrsfs{A}(Q,\Sigma,\delta)$, where $Q$ is the set of \emph{states}, $\Sigma$ is the input \emph{alphabet}, and $\delta\colon Q \times \Sigma \to Q$ is the \emph{transition function}.
We extend $\delta$ to the function $Q \times \Sigma^* \to Q$ in the usual way.
Throughout the paper, by $n$ we denote the number of states $|Q|$.

By $\Sigma^{\le i}$ we denote the set of all words over $\Sigma$ of length at most $i$.
Given a state $q \in Q$ and a word $w \in \Sigma^*$ we write shortly $q\cdot w = \delta(q,w)$.
Given a subset $S \subseteq Q$ we write $S\cdot w$ for the image $\{q\cdot w \mid q \in S\}$.
Then, $S\cdot w^{-1}$ is the preimage $\{q \in Q \mid q\cdot w \in S\}$, and when $S$ is a singleton we also write $q\cdot w^{-1} = \{q\}\cdot w^{-1}$.

The \emph{rank} of a word $w \in \Sigma^*$ is the cardinality of the image of $Q$ under the action of this word: $|Q\cdot w|$.
A word is \emph{reset} or \emph{synchronizing} if it has rank $1$.
An automaton is \emph{synchronizing} if it admits a reset word.
The \emph{reset threshold} $\rt(\mathrsfs{A})$ is the length of the shortest reset words.

We say that a word $w \in \Sigma^*$ \emph{compresses} a subset $S \subseteq Q$ if $|S\cdot w| < |S|$.
A word $w \in \Sigma^*$ \emph{avoids} a state $q \in Q$ if $q \notin Q\cdot w$.
A state that admits an avoiding word is \emph{avoidable}.
We also say that a state $q$ is \emph{avoidable from a subset} $S$ if there exists a word $w$ such that $q \notin S\cdot w$.

The famous \v{C}ern\'{y} conjecture, formally formulated in 1969, is one of the most longstanding open problems in automata theory.
It states that every synchronizing $n$-state automaton has a reset word of length at most $(n-1)^2$.
This bound would be tight, since it is reached for every $n$ by the \v{C}ern\'{y} automata \cite{Cerny1964}.
Fig.~\ref{fig:cerny4} shows the \v{C}ern\'{y} automaton with $n=4$ states.
Its shortest reset word is $b a^3 b a^3 b$.

\begin{figure}[htb]
\centering\includegraphics{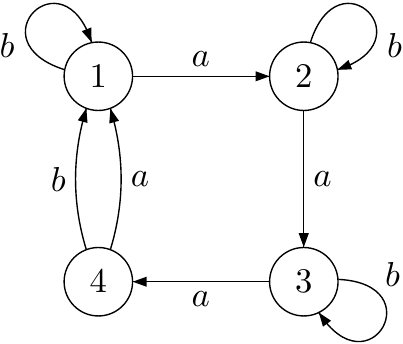}
\caption{The \v{C}ern\'{y} automaton with $4$ states.}\label{fig:cerny4}
\end{figure}

The first general upper bound for the reset threshold given by \v{C}ern\'{y} in \cite{Cerny1964} was $2^n-n-1$.
Later, it was improved several times:
$\frac{1}{2}n^3 - \frac{3}{2}n^2 + n + 1$ given by Starke \cite{Starke1966} in~1966,
$\frac{1}{3}n^3 - \frac{3}{2}n^2 + 25/6 n - 4$ by \v{C}ern\'{y}, Pirick\'{a}, and Rosenauerov\'{a} \cite{CernyPirickaRosenauerova1971} in~1971,
$\frac{7}{27}n^3 - 17/18 n^2 + 17/6n - 3$ by Pin \cite{Pin1978SurLesMots} in~1978, and
$(\frac{1}{2}-\frac{\pi}{36})n^3 + o(n^3)$ by Pin \cite{Pin1983OnTwoCombinatorialProblems} in~1981.

Then, the well known upper bound was established in 1982 by Pin and Frankl through the following combinatorial theorem:
\begin{theorem}[\textrm{\cite{Fr1982,Pin1983OnTwoCombinatorialProblems}}]\label{thm:pin_bound}
Let $\mathrsfs{A}(Q,\Sigma,\delta)$ be a strongly connected synchronizing automaton, and
consider a subset $S \subseteq Q$ of cardinality $\ge 2$.
Then there exists a word such that $|S\cdot w| < |S|$ of length at most
$$\frac{(n-|S|+2)\cdot(n-|S|+1)}{2}.$$
\end{theorem}
For integers $1 \le i,j \le n$ we define
\begin{eqnarray*}
C(j,i) & = & \sum_{s=i+1}^{j} \frac{(n-s+2)\cdot(n-s+1)}{2}.
\end{eqnarray*}
From Theorem~\ref{thm:pin_bound}, $C(j,i)$ is an upper bound on the length of the shortest words compressing a subset of size $j$ to a subset of size at most $i$: starting from a subset $S$ of size $j$, we iteratively apply Theorem~\ref{thm:pin_bound} to bound the length of a shortest word compressing each (in the worst case) of the obtained subsets of sizes $j,j-1,\ldots,i+1$.
This yields the well known bound on the length of the shortest reset words:
$$\rt(\mathrsfs{A}) \le C(n,1) = \frac{n^3-n}{6}.$$
This bound was also discovered independently in~\cite{KRS1987}.
Actually, the best bound was $\frac{n^3-n}{6}-1$ (for $n \ge 4$), since Pin~\cite{Pin1983OnTwoCombinatorialProblems} proved that (for $n \ge 4$) there is a word compressing $Q$ to a subset of size $n-3$ by a word of length $9$ (instead of $10$).
Theorem~\ref{thm:pin_bound} also bounds the lengths of a compressing word found by a greedy algorithm (e.g.~\cite{AG2016GreedyAlgorithms,Ep1990}), which is an algorithm finding a reset word by iterative application of a shortest word compressing the current subset.
For about 35 years, there was no progress in improving the bound in the general case.

However, better bounds have been obtained for a lot of special classes of automata, for example for
oriented (monotonic) automata \cite{Ep1990}, circular automata \cite{Dubuc1998}, Eulerian automata \cite{Kari2003Eulerian}, aperiodic automata \cite{Tr2007Aperiodic}, generalized and weakly monotonic automata \cite{AV2005SynchronizingGeneralizedMonotonicAutomata,Volkov2009ChainOfPartialOrders}, automata with a sink (zero) state \cite{Ma2008ZeroState}, one-cluster automata \cite{BBP2011QuadraticUpperBoundInOneCluster,Steinberg2011OneClusterPrime}, quasi-Eulerian and quasi-one-cluster automata \cite{Berlinkov2013QuasiEulerianOneCluster}, automata respecting intervals of a directed graph \cite{GK2013AutomataRespectingIntervals}, decoders of finite prefix codes \cite{BS2016AlgebraicSynchronizationCriterion,BiskupPlandowski2009HuffmanCodes}, automata with a letter of small rank \cite{BS2016AlgebraicSynchronizationCriterion,Pin1972Utilisation}, and a subclass of 1-contracting automata \cite{Don2016CernyAnd1Contracting}.
See also \cite{Volkov2008Survey} for a survey.

In~2011, Trahtman claimed the better upper bound $(7n^3+6n-16)/48$ \cite{Tr2011ModifyingUpperBound}.
Unfortunately, the proof contains an error, and so the result remains unproved.
The idea was to utilize avoiding words; \cite[Lemma~3]{Tr2011ModifyingUpperBound} states that for every $q \in Q$ there exists an avoiding word of length at most $n-1$.
A counterexample to this was found in~\cite{GJT2014ANoteOnARecentAttempt}, where it was also suggested that providing any linear upper bound on the length of avoiding words would also imply an improvement for the upper bound on the reset threshold.

The avoiding word problem is similar to synchronization: instead of bringing the automaton into one state, we ask how long word we require to not being in a particular state.
For the automaton from Fig.~\ref{fig:cerny4}, the shortest avoiding words for states $1$, $2$, $3$, $4$ are $ba$, $baa$, $baaa$, and $b$, respectively.
So far, only a trivial cubic upper bound $\rt(\mathrsfs{A})+1$ was known for synchronizing automata.
Avoiding words do not necessarily exist in general, but they always do for every state in the case of a synchronizing automaton unless there is a \emph{sink state} (\cite{Ma2008ZeroState}), for which all letters act like identity.

The main contributions in this paper are as follows:
We prove upper bounds on the length of the shortest avoiding words, in particular the quadratic bound $(n-1)(n-2)+2$.
Also, the length of avoiding words is connected with the length of compressing words.
We show that for every state $q$ and a subset of states $S$, either there is a short avoiding word for $q$ from $S$ or a short compressing word for $S$.
This connection leads to the main idea for the improvement of the general upper bound on the reset threshold: either improve by avoiding words, or use shorter compressing words directly to reduce the bound obtained by Theorem~\ref{thm:pin_bound}.
In contrast to the previous approaches, which bounded the length of the compressing words independently for each size $|S|$, the new bound utilizes a conditional approach.

The new upper bound is
$$(85059 n^3 + 90024 n^2 + 196504 n - 10648)/511104,$$
which is slightly better than the much simpler formula $114 n^3 / 685 + O(n^2)$.
The latter improves the coefficient of $n^3$ by $1/4110$.
In the last section we discuss open problems and further possibilities for improvements.

\section{Avoiding words}

For the next lemma, we need to introduce a few definitions from linear algebra for automata (see, e.g.,~\cite{BS2016AlgebraicSynchronizationCriterion,Kari2003Eulerian,Pin1972Utilisation}).
By $\mathbb{R}^n$ we denote the real $n$-dimensional linear space of row vectors.
Without loss of generality we assume that $Q = \{1,2,\ldots,n\}$.
For a vector $v \in \mathbb{R}^n$, we denote the value at an $i$-th position by $v(i)$.
For a subset $S \subseteq Q$, by $[S]$ we denote its characteristic row vector, which has $[S](i)=1$ if $i \in S$, and $[S](i)=0$ otherwise.
Similarly, for a matrix $M$, we denote the value at an $i$-th row and a $j$-th column by $M(i,j)$.
For a word $w \in \Sigma^*$, by $[w]$ we denote the $n \times n$ matrix of the transformation of $w$:
$[w](i,j) = 1$ if $i\cdot w=j$ (state $i$ is mapped to state $j$ by the transformation of $w$), and $[w](i,j)=0$ otherwise.

Right matrix multiplication corresponds to concatenation of two words; i.e.\ for every two words $u,v \in \Sigma^*$ we have $[uv] = [u]\cdot[v]$.
For a subset $S$ we have $([S][u])(i)$ equal to the number of states from $S$ mapped by the transformation of $u$ to state $i$.
In particular, $([S][u])(i) \ge 1$ if and only if $[S\cdot u](i) = 1$.
Note that for $w \in \Sigma^*$, the matrix $[w]$ contains exactly one $1$ in each row.
Therefore, these are \emph{stochastic} matrices, and we have the property that for any $v \in \mathbb{R}^n$, right matrix multiplication by $[w]$ preserves the sum of the entries, i.e.\ $\sum_{i \in Q} [v](i) = \sum_{i \in Q} ([v][w])(i)$.

For example, for the automaton from Fig.~\ref{fig:cerny4} we have:
$$[a]=\left(
  \begin{smallmatrix}
    0 & 1 & 0 & 0\\
    0 & 0 & 1 & 0\\
    0 & 0 & 0 & 1\\
    1 & 0 & 0 & 0
  \end{smallmatrix}
\right),\ [b]=\left(
  \begin{smallmatrix}
    1 & 0 & 0 & 0\\
    0 & 1 & 0 & 0\\
    0 & 0 & 1 & 0\\
    1 & 0 & 0 & 0
  \end{smallmatrix}
\right),\ [ba]=\left(
  \begin{smallmatrix}
    0 & 1 & 0 & 0\\
    0 & 0 & 1 & 0\\
    0 & 0 & 0 & 1\\
    0 & 1 & 0 & 0
  \end{smallmatrix}
\right).$$
If $[S] = [1,0,1,1]$, then $[S][ba]=[S][b][a] = [0,2,0,1]$.

The linear subspace \emph{spanned} by a set of vectors $V$ is denoted by $\lspan(V)$.
Given a linear subspace $L \subseteq \mathbb{R}^n$ and an $n \times n$ matrix $m$, the linear subspace mapped by $m$ is $L m = \{v m \mid v \in L\}$.
The \emph{dimension} of a linear subspace $L$ is denoted by $\dim(L)$.

The following key lemma states that by a short (linear) word we can either avoid a state (or one of the states from some set $A$) from the current subset or compress the current subset.
\begin{lemma}\label{lem:avoid_or_compress}
Let $\mathrsfs{A}(Q,\Sigma,\delta)$ be an $n$-state automaton.
Consider a non-empty subset $S \subseteq Q$ and a non-empty proper subset $A \subsetneq S$.
Suppose that there is a word $w \in \Sigma^*$ such that $A \nsubseteq S\cdot w$.
Then there exists a word $w$ length at most $n-|A|$ satisfying either 
\begin{enumerate}
\item $A \nsubseteq S\cdot w$, or
\item $|S\cdot w| < |S|$.
\end{enumerate}
\end{lemma}
\begin{proof}
Let $L_i = \lspan(\{[S][w] \mid w \in \Sigma^{\le i}\})$.
We consider the following sequence of linear subspaces:
$$L_0 \subseteq L_1 \subseteq L_2 \subseteq \ldots,$$
and use the ascending chain condition (see, e.g.,~\cite{BS2016AlgebraicSynchronizationCriterion,Kari2003Eulerian,Pin1972Utilisation,Steinberg2011AveragingTrick}):
\begin{itemize}
\item If $L_k = L_{k+1}$, then we claim that also $L_{k+1} = L_{k+2} = \ldots$ holds.
Observe that for all $i \ge 0$ we have:
$$L_{i+1} = \lspan\left(L_i \cup \bigcup_{a \in \Sigma} L_i [a]\right).$$
Hence, if $L_k = L_{k+1}$, then for $i=k$ we obtain
$$L_{k+1} = \lspan\left(L_{k+1} \cup \bigcup_{a \in \Sigma} L_{k+1} [a]\right) = L_{k+2},$$
and so $L_{k+i} = L_k$ for all $i \ge 0$.
\item Let $i$ be the smallest integer such that $L_i = L_{i+1}$.
Then $m=\dim(L_i)$ is the maximum among the dimensions of the subspaces from the above sequence.
\item $\dim(L_0) = 1$ and the dimensions grow by at least $1$ up to $m$.
Hence, we have
$$\dim(L_{n-|A|}) \ge \min\{m,n-|A|+1\}.$$
\end{itemize}

Note that if for a word $w$ the vector $v=[S][w]$ has $v(q)=0$ for some $q \in A$, then $q \notin S\cdot w$, and we have Case~(1).
If $v=[S][w]$ has $v(q) \ge 2$ for some $q \in A$, then a pair of states from $S$ is compressed by the action of $w$ (to state $q$), and we have Case~(2).

Now, we show that in the spanning set of $L_{n-|A|}$ there must be a vector that contains either $0$ or an integer $\ge 2$ at the position corresponding to a state from $A$, which implies that there exists a word $w$ of length at most $n-|A|$ satisfying either Case~(1) or Case~(2).
Suppose for a contradiction that this is not the case.
Every vector $v \in L_k$ is a linear combination of the vectors from the spanning set; let $c$ be the sum of the coefficients of the spanning vectors in such a linear combination.
Every vector $[S][w]$ in the spanning set has the sum of elements equal to $|S|$ and has $1$ at all the positions corresponding to the states from $A$.
Hence, the sum of the entries in $v$ is equal to $c|S|$, and at every position corresponding to the states from $A$ we have value $c$.
The sum of the entries at the positions corresponding to the states from $Q \setminus A$ equals $c(|S|-|A|)$.
Therefore, every $q \in A$ satisfies the following equality:
$$v(q) = \frac{1}{|S|-|A|}\cdot\sum_{p \in Q \setminus A} v(p).$$
It follows that the values at the positions corresponding to the states from $A$ are completely determined by the sum of the values from the other positions, which means that the dimension of $L_{n-|A|}$ is at most $n-|A|$.
We assumed in the lemma that there exists a word $w$ avoiding a state from $A$.
Hence, $[S][w]$ has $0$ at some position corresponding to a state from $A$, and therefore breaks the above equality for this state, as the right side is non-zero.
Therefore, the subspace $L_{|w|}$ must have a larger dimension that the dimension of $\dim(L_{n-|A|})$.
This means that the dimension of $L_{n-|A|}$ is not maximal, which contradicts $\dim(L_{n-|A|}) \ge \min\{m,n-|A|+1\}$.
\end{proof}

Lemma~\ref{lem:avoid_or_compress} can be applied iteratively to obtain a word compressing the given subset to the desired size.
\begin{lemma}\label{lem:avoiding_iterations}
Let $\mathrsfs{A}(Q,\Sigma,\delta)$ be an $n$-state automaton.
Consider a non-empty subset $S \subseteq Q$ and a non-empty proper subset $A \subsetneq S$.
Let $k \ge 1$ be an integer.
Suppose that there exists a word $w \in \Sigma^*$ such that $A \nsubseteq S\cdot w$.
Then there is a word $w$ of length at most $k(n-|A|)$ satisfying either:
\begin{enumerate}
\item $A \nsubseteq S\cdot w$, or
\item $|S\cdot w| \le |S|-k$.
\end{enumerate}
\end{lemma}
\begin{proof}
If Case~(1) holds for some $w \in \Sigma^{\le k(n-|A|)}$ then we are done; suppose this is not the case.

We iteratively apply Lemma~\ref{lem:avoid_or_compress} $k$ times for subset $A$ starting from subset $S$:
For $i=1,\ldots,k$ we apply the lemma for the subset $S\cdot w_1 \dots w_{i-1}$, where $w_j \in \Sigma^{\le n-|A|}$ is the word obtained from the lemma in the $j$-th iteration.

In every iteration, we must get Case~(2) of Lemma~\ref{lem:avoid_or_compress} ($|S\cdot w| < |S|$), as otherwise $A \nsubseteq S\cdot w_1 \dots w_i$, which contradicts our assumption that Case~(1) does not hold for every word of length at most $k(n-|A|) \ge i(n-|A|)$.
Also, for $i \le k-1$, we must have $A \subset S\cdot w_1 \dots w_i$ (i.e.\ $A$ is a proper subset); otherwise $A \nsubseteq S\cdot w_1 \dots w_i a$ for some letter $a \in \Sigma$ as $A$ contains a state that can be avoided from $S$, and this word has length at most $k(n-|A|)$ which again contradicts our assumption.
Therefore, the conditions are met for every iteration so we can apply the lemma $k$ times.

It follows that the obtained word $w_1 \dots w_k$ is such that $|S\cdot w_1 \dots w_k| \le |S|-k$.
\end{proof}

If the subset $A$ of states to avoid is large, the following approach can lead to a better bound:
\begin{lemma}\label{lem:avoiding_iterations2}
Let $\mathrsfs{A}(Q,\Sigma,\delta)$ be an $n$-state automaton.
Consider a non-empty subset $S \subseteq Q$ and a non-empty subset $A \subseteq S$.
If there exists a word $w \in \Sigma^*$ such that $A \nsubseteq S\cdot w$,
then there exists such a word of length at most $(|S|-|A|)(n-|A|)+1$.
\end{lemma}
\begin{proof}
As in the proof of Lemma~\ref{lem:avoiding_iterations}, we iteratively apply Lemma~\ref{lem:avoid_or_compress} at most $|S|-|A|$ times for subset $A$ starting from subset $S$, stopping if the conditions are not met.
It is possible that we do not do any iteration, which is the case when $A=S$.

In every iteration, we obtain a word $w_i$ of length at most $n-|A|$.
If we get $A \nsubseteq S\cdot w_1 \dots w_i$ in some $i$-th iteration, then we are done as the word $w_1 \dots w_i$ has length at most $(|S|-|A|)(n-|A|)$.

If we get $A = S\cdot w_1 \dots w_i$ for some $i \in \{0,\ldots,|S|-|A|\}$, then observe that there must exist a letter $a \in \Sigma$ such that $A\cdot a \neq A$, because $A$ contains an avoidable state from $S \supseteq A$.
Note that since $|S\cdot w_1 \dots w_i|<|S\cdot w_1 \dots w_{i-1}|$ for every $i=1,\ldots,k$, after the $(|S|-|A|)$-th iteration we must have $|S\cdot w_1 \dots w_k| \le |S|-(|S|-|A|)=|A|$, we must get this case after the last iteration.
It follows that in any case we obtain the word $w_1 \dots w_i a$ of length at most $(|S|-|A|)(n-|A|)+1$.
\end{proof}

We state a quadratic upper bound on the length of the shortest avoiding words:
\begin{corollary}\label{cor:avoiding_bound}
For $n \ge 2$, in an $n$-state automaton $\mathrsfs{A}(Q,\Sigma,\delta)$, for every non-empty proper subset $A \subset Q$ containing an avoidable state, there exists a word avoiding a state from $A$ of length at most
$$(n-1-|A|)(n-|A|)+2.$$
\end{corollary}
\begin{proof}
Since there exists an avoidable state in $A$, there is a letter $a \in \Sigma$ such that $|Q\cdot a| < n$.

If $A \nsubseteq Q\cdot a$ then we are done with a word of length $1$.
Otherwise $A \subseteq Q\cdot a$, so we use Lemma~\ref{lem:avoiding_iterations2} with subset $A$ and subset $S=Q\cdot a$.
Since there exists a word avoiding a state from $A$, the lemma yields a word $w$ of length at most $(|S|-|A|)(n-|A|)+1 \le (n-1-|A|)(n-|A|)+1$.
Thus, $aw$ avoids a state from $A$ and has length at most $(n-1-|A|)(n-|A|)+2$.
\end{proof}
In particular, we obtain the upper bound $(n-2)(n-1)+2$ on the length of the shortest avoiding words for any state ($|A|=1$).

\begin{theorem}\label{thm:polynomial_time}
The words from Lemma~\ref{lem:avoid_or_compress}, Lemma~\ref{lem:avoiding_iterations}, Lemma~\ref{lem:avoiding_iterations2}, and Corollary~\ref{cor:avoiding_bound} can be found in polynomial time.
\end{theorem}
\begin{proof}
We use the reduction procedure from~\cite{BS2016AlgebraicSynchronizationCriterion}, which in polynomial time replaces each set $\Sigma^{\le i}$ in the proof of Lemma~\ref{lem:avoid_or_compress} with a set $W_i$ containing at most $i+1$ words such that $L_i$ has the same dimension.

The procedure starts for $i=0$ with $\{\varepsilon\}$ (the set with the empty word) and inductively constructs a set $W_i$ assuming we have found $W_{i-1}$. This is done by considering all words $wa$ for $w \in W_{i-1}$ and $a \in \Sigma$ and setting $W_i = W_{i-1} \cup \{wa\}$ for which the dimension of the corresponding subspace grows. There always exists such a word $wa$, which is argued by ascending chain condition.

Then, the set $W_m$ is used to span the first linear subspace with the maximal dimension ($L_m$), so we can find a word satisfying Case~(1) or Case~(2) of Lemma~\ref{lem:avoid_or_compress} in $W_m$.
It is obvious that the corresponding words from the other proofs are constructible in polynomial time.
\end{proof}

\section{Improved bound on reset threshold}

In this section, we consider a synchronizing $n$-state automaton $\mathrsfs{A}(Q,\Sigma,\delta)$.
Obviously, in such an automaton, every state is avoidable unless there is a sink state (a state $q$ such that $q\cdot a = q$ for all $a \in \Sigma$), which cannot be avoided.
For synchronizing automata with a sink state the tight upper bound is $n(n-1)/2$ (see, e.g., \cite{Rystsov1997ResetWordsForCummutativeAndSolvableAutomata}).
Thus we can assume that $\mathrsfs{A}$ does not have a sink state, and so Lemma~\ref{lem:avoid_or_compress} and Lemma~\ref{lem:avoiding_iterations} can be applied for every non-empty subset $A$.

\begin{lemma}\label{lem:preimage_states}
Let $w \in \Sigma^*$ and let $g = \min\{|q \cdot w^{-1}| \mid q \in Q\cdot w\}$.
There are at least $(g+1)|Q\cdot w|-n$ states $q \in Q\cdot w$ such that $|q\cdot w^{-1}| = g$.
\end{lemma}
\begin{proof}
Let $d$ be the number of states $q \in Q\cdot w$ whose preimages under $w^{-1}$ have size equal to $g$.
So $|Q\cdot w|-d$ states have the preimages of size at least $g+1$.
Note that $(Q\cdot w)\cdot w^{-1}=Q$, and that the sets $q\cdot w^{-1}$ and $p\cdot w^{-1}$ are disjoint for all pairs of states $q \neq p$.
So $Q\cdot w^{-1}$ has cardinality at least $dg+(g+1)(|Q\cdot w|-d)=(g+1)|Q\cdot w|-d$.
Since this cannot be larger than $n=|Q|$, we get $d \ge (g+1)|Q\cdot w|-n$.
\end{proof}
From Lemma~\ref{lem:preimage_states}, in particular, we get that there are at least $2|Q\cdot w|-n$ states in the image $Q\cdot w$ with a unique state in the preimage.

The following lemma is based on~\cite[Lemma~4]{Tr2011ModifyingUpperBound}, but with a more general bound:
\begin{lemma}\label{lem:using_avoiding}
Let $w \in \Sigma^*$ be a word of rank $r \ge \lfloor(n+1)/2\rfloor$.
Suppose that for some integer $k \ge 1$, for every $A \subset Q$ of size $1 \le |A| \le n-1$, there is a word $v_A \in \Sigma^{\le k(n-|A|)}$ such that $A \nsubseteq Q \cdot v_A$.
Then there is a word of rank at most $n/2$ and length at most
$$|w| + k\frac{n^2 - (2 n - 2 r - 1)^2}{4}.$$
\end{lemma}
\begin{proof}
For $i=r,r-1,\ldots,\lfloor n/2\rfloor$, we inductively construct words $w_i$ of length $\le |w| + k (r - i) (2n - r - i - 1)$ of rank at most $i$.
First, let $w_r = w$.

Let $i < r$ and suppose that we have already found $w_{i+1}$.
If already $|Q\cdot w_{i+1}| \le i$ then we just set $w_i = w_{i+1}$.
Otherwise, we have $|Q\cdot w_{i+1}| = i+1$.

Because $i+1 \ge (n+1)/2$, there exists a non-empty subset of $Q\cdot w_{i+1}$ of states with a unique state in the unique preimage.
By Lemma~\ref{lem:preimage_states}, we let $X \subseteq Q\cdot w_{i+1}$ to be a subset of size $2|Q\cdot w_{i+1}|-n = 2i+2-n$ of states $q \in Q\cdot w_{i+1}$ such that $|q\cdot w_{i+1}^{-1}|=1$.
We set $w_i = v_X w_{i+1}$, where $v_X$ is the avoiding word from the assumption of the lemma for set $X$.
We have $p \notin Q \cdot v_X$ for some $p \in X$.

State $p$ is the only state mapped by the transformation of $w_{i+1}$ to some state $q = p\cdot w_{i+1}$, i.e.\ there is no other state $p'$ such that $p'\cdot w_{i+1} = q$.
Hence we know that $q \notin Q \cdot w_i = Q \cdot v_X w_{i+1}$.
Since $Q \cdot w_i \subseteq Q \cdot w_{i+1}$, $q \notin Q \cdot w_i$ but $q \in Q \cdot w_{i+1}$, we have $Q \cdot w_i \subsetneq Q \cdot w_{i+1}$.
Therefore, we have rank
$$|Q \cdot w_i| \le |Q \cdot w_{i+1}|-1 \le i+1-1 = i,$$
and length
\begin{align*}
|w_i| \le\ & k(n-|A|) + |w_{i+1}| \\
      \le\ & 2k(n-i-1) + k (r - (i+1)) (2n - r - (i+1) - 1) + |w| \\
        =\ & k (r - i) (2n - r - i - 1) + |w|.
\end{align*}

Finally, for $i=\lfloor n/2\rfloor$ we obtain:
\begin{align*}
     & |w| + k (r - \lfloor n/2\rfloor) (2n - r - \lfloor n/2\rfloor - 1) \\
\le\ & |w| + k (r - (n-1)/2) (2n - r - (n-1)/2 - 1) \\
  =\ & |w| + k (n^2 - (2 n - 2 r - 1)^2)/4.
\end{align*}
\end{proof}

Note that Lemma~\ref{lem:avoiding_iterations2} also provides an upper bound on the length of the shortest avoiding words, but it is larger than that the corresponding bound from Theorem~\ref{thm:pin_bound}, and so would not yield an improvement when used as in Lemma~\ref{lem:using_avoiding}.
Therefore, we use there an assumption about the length of the shortest avoiding words.

We observe that it is profitable to use Theorem~\ref{thm:pin_bound} to find the starting word $w$, as long as $C(i+1,i)$ is smaller than $k(n-|A|)$.
An approximate solution is to find the starting word $w$ of rank at most $n-4k$.
The following lemma utilizes this idea.
\begin{lemma}\label{lem:using_compressing_and_avoiding}
Suppose that for some integer $k$, $1 \le k \le n/8$, for every $A \subset Q$ of size $1 \le |A| \le n-1$, there is a word $v_A \in \Sigma^{\le k(n-|A|)}$ such that $A \nsubseteq Q \cdot v_A$.
Then there is a word of rank at most $n/2$ and length at most
$$k\frac{3 n^2 - 64 k^2 + 144 k + 13}{12}.$$
\end{lemma}
\begin{proof}
From Theorem~\ref{thm:pin_bound}, let $w$ be a word of rank at most $n-4k$ and length at most
$$C(n,n-4k) = 4k (8 k^2 + 6 k + 1) / 3.$$

If $w$ has rank $\ge \lfloor(n+1)/2\rfloor$, then we apply Lemma~\ref{lem:using_avoiding} and obtain a word of rank at most $n/2$ and length at most
\begin{align*}
   & \frac{4k (8 k^2 + 6 k + 1)}{3} + \frac{k (n^2 - (2 n - 2 (n-4k) - 1)^2)}{4} \\
=\ & \frac{k (3 n^2 - 64 k^2 + 144 k + 13)}{12}.
\end{align*}
Otherwise, $w$ has rank $< n/2$, and because
$$k(n^2 - (2 n - 2 (n-4k) - 1)^2)/4 = k(n^2 - (8k-1)^2)/4$$
is positive for $1 \le k \le n/8$ (and $n \ge 8$),
the upper bound is also valid.
Thus, $w$ has the desired length.
\end{proof}

We prove a parametrized upper bound on the reset threshold, depending on whether the assumption in Lemma~\ref{lem:using_compressing_and_avoiding} holds.
When the assumption holds, the lemma provides an upper bound using avoiding words; otherwise, we have a quadratic word of a particular rank that yields an improvement.
\begin{lemma}\label{lem:parametrized_bound}
For every integer $1 \le k \le n/8$, there exists a reset word of length at most
\begin{align*}
   & \max\left\{k\frac{3 n^2 - 64 k^2 + 144 k + 13}{12},\; k(n-1)+C(n-k,\lfloor n/2 \rfloor)\right\} \\
+\ & C(\lfloor n/2\rfloor,1).
\end{align*}
\end{lemma}
\begin{proof}
We use Lemma~\ref{lem:avoiding_iterations} with the given $k$ and subset $S=Q$.

Suppose that Case~(1) from Lemma~\ref{lem:avoiding_iterations} holds for every $A \subset Q$ with $1 \le |A| \le n-1$.
Then by Lemma~\ref{lem:using_compressing_and_avoiding} we obtain a word $w$ of rank $\le n/2$ and length $\le k(3 n^2 - 64 k^2 + 144 k + 13)/12$.

Suppose that Case~(2) from Lemma~\ref{lem:avoiding_iterations} holds for some $A \subset Q$ with $1 \le |A| \le n-1$.
Then we have a word $w$ of rank $\le n-k$ and length $\le k(n-1)$.
By Theorem~\ref{thm:pin_bound}, we construct a word compressing $Q\cdot w$ to a subset of size $\le n/2$.
Then $k(n-1)+C(n-k,\lfloor n/2 \rfloor)$ is an upper bound for the length of the found word of rank $\le n/2$.

Finally, we need to take the maximum from both cases, and add $C(\lfloor n/2\rfloor,1)$ to bound the length of a word compressing a subset of size $\lfloor n/2\rfloor$ to a singleton.
\end{proof}

Now, by finding a suitable $k$, we state the new general upper bound on the reset threshold:
\begin{theorem}\label{thm:new_bound}
$$\rt(\mathrsfs{A}) \le (85059 n^3 + 90024 n^2 + 196504 n - 10648)/511104.$$
\end{theorem}
\begin{proof}
We use Lemma~\ref{lem:parametrized_bound} with a suitable $k$ that minimizes the maximum for large enough $n$.

First, we bound $C(n-k,\lfloor n/2 \rfloor)$ in the second argument in the maximum.
If $n$ is even then
\begin{align*}
C(n-k,\lfloor n/2 \rfloor) =\ & C(n-k,n/2) \\
=\ & \sum_{s=n/2+1}^{n-k} \frac{(n-s+2)(n-s+1)}{2} \\
=\ & \frac{n^3 + 6 n^2 + 8 n - 8 k^3 - 24 k^2 - 16 k}{48}.
\end{align*}
If $n$ is odd then
\begin{align*}
C(n-k,\lfloor n/2 \rfloor) =\ & C(n-k,(n-1)/2) \\
=\ & \sum_{s=(n-1)/2+1}^{n-k} \frac{(n-s+2)(n-s+1)}{2} \\
=\ & \frac{n^3 + 9 n^2 + 23 n - 8 k^3 - 24 k^2 - 16 k + 15}{48},
\end{align*}
which is larger than the previous one.

Now we discuss our choice of $k$; any value of $k$ gives a bound but we try to get it minimal.
Assume that $n$ is large enough.
Note that for the largest possible value $k=n/8$ the first function in the maximum from Lemma~\ref{lem:parametrized_bound} yields the coefficient of $n^3$ equal to $1/48$ (the same as by $C(n,\lfloor n/2\rfloor)$), hence does not give an improvement.
For a similar reason, we reject small values $k \in o(n)$.
Within linear values $k$ of $n$, the first function decreases and the second function increases with $k$.
Since they are continuous, it is enough to consider the values of $k$ such that both functions are equal. 
The approximate solution is $k \simeq 0.11375462 n$.
For simplicity of the calculations and the final formula, we use the approximation $k = \lfloor 5/44 n\rfloor$.
Note that any value of $k$ within the valid range will lead to a correct bound, and we use $5/44$ since it is the best approximation by a rational number using integers with at most two digits.

We assume $n \ge 9$; for the smaller values of $n$ the bound is a valid upper bound since it gives larger values than the bound from Theorem~\ref{thm:pin_bound}.

In the following calculations, we use the fact that $5/44n-1 < \lfloor 5/44 n\rfloor$ and $5/44n-1$ is non-negative.
By substitution, for the first function in the maximum we have
\begin{align}
   & k\frac{3 n^2 - 64 k^2 + 144 k + 13}{12} \nonumber\\
<\ & (5/44 n)\frac{3 n^2 - 64 (5/44 n - 1)^2 + 144 (5/44 n) + 13}{12} \nonumber\\
=\ & (5 n (263 n^2 + 3740 n - 6171))/63888,
\end{align}
and for the second function we have
\begin{align}
   & k(n-1) + \frac{n^3 + 9 n^2 + 23 n - 8 k^3 - 24 k^2 - 16 k + 15}{48} \nonumber\\
<\ & (5/44 n)(n-1) + \big(n^3 + 9 n^2 + 23 n - 8 (5/44 n - 1)^3 \nonumber\\
   & - 24 (5/44 n - 1)^2 - 16 (5/44 n - 1) + 15\big)/48 \nonumber\\ 
=\ & (10523 n^3 + 153912 n^2 + 196504 n + 159720)/511104.
\end{align}
Note that (2) is larger than (1) for all $n$.

Now we have to bound $C(\lfloor n/2\rfloor,1)$.
If $n$ is even then
$$C(\lfloor n/2\rfloor,1) = C(n/2,1) = (7 n^3 - 6 n^2 - 16)/48.$$
If $n$ is odd then
$$C(\lfloor n/2\rfloor,1) = C((n-1)/2,1) = (7 n^3 - 9 n^2 - 31 n - 15)/48,$$
which is smaller than the previous one for $n \ge 2$.

Finally, we obtain
\begin{align*}
   & \frac{10523 n^3 + 152262 n^2 + 189244 n + 191664}{511104} + \frac{7 n^3 - 6 n^2 - 16}{48} \\
=\ & \frac{85059 n^3 + 90024 n^2 + 196504 n - 10648}{511104}.
\end{align*}
\end{proof}
The theorem improves the old well known bound $(n^3-n)/6-1$ by the factor $85059/85184$, or by the coefficient $125/511104$ of $n^3$.
This is slightly better than the simpler formula $114 n^3 / 685 + O(n^2)$.

The bound does not necessarily apply for the words obtained by a greedy compression algorithm for synchronization (\cite{AG2016GreedyAlgorithms,Ep1990}), because the words in the proof of Lemma~\ref{lem:using_avoiding} are constructed by appending avoiding words at the beginning.
However, we can show that there exists a polynomial algorithm finding words of lengths within the bound.
\begin{proposition}
A reset word of length within the bound from Theorem~\ref{thm:new_bound} can be computed in polynomial time.
\end{proposition}
\begin{proof}
We use $k$ from the proof of Theorem~\ref{thm:new_bound}.
We follow the construction from the proof of Lemma~\ref{lem:using_avoiding}.
By Theorem~\ref{thm:polynomial_time}, we can compute a word from Lemma~\ref{lem:avoid_or_compress} for a subset $A$.
If (1) holds every time, then we use the obtained word from Lemma~\ref{lem:using_avoiding}.
Otherwise, we use the word from Lemma~\ref{lem:avoid_or_compress} for which (2) holds.
Finally, the words of lengths at most $C(j,i)$ are computed using a greedy compression algorithm (\cite{AG2016GreedyAlgorithms}).
\end{proof}

\section{Further remarks and open problems}

Although the improvement in terms of the cubic coefficient is small, it breaks longstanding persistence of the old bound from \cite{Pin1983OnTwoCombinatorialProblems}, and possibly opens the area for further progress.

Tiny improvements of the bound from Theorem~\ref{thm:new_bound} are possible with more effort yielding better calculations, for example by tuning the value of $k$ in Theorem~\ref{thm:new_bound}, better rounding, using better bounds at the beginning (note that one can find a shorter word than the word of rank $k$ when Case~(2) holds in Lemma~\ref{lem:avoiding_iterations} by combining with Theorem~\ref{thm:pin_bound}).
These however do not add new ideas.

We present a few open problems that could help to understand avoiding words better, and maybe lead to further improvements.

\medskip
\noindent\textbf{$\bullet$ Avoiding a state}:
The first natural possibility for improving the bound is to show a better bound on the length of the shortest avoiding words.
For strongly connected synchronizing automata, currently the best known lower bound is $2n-3$ by Vojt{\v e}ch Vorel\footnote{personal communication, unpublished, 2016} (binary series), whereas $2n-2$ is conjectured to be a tight upper bound based on experiments \cite{KKS2016ExperimentsWithSynchronizingAutomata}.
\begin{openproblem}
Is $2n-2$ the tight upper bound on the length of the shortest avoiding words for a single state?
\end{openproblem}

\medskip
\noindent\textbf{$\bullet$ Avoiding a subset}:
The technique from Lemma~\ref{lem:using_avoiding} can be applied only for compressing $Q$ to a subset of size at most $n/2$, because at this point there can be no states with a unique state in the preimage.
To bypass this obstacle, we can generalize the concept of avoiding to subsets, and say that a word $w$ \emph{avoids} a subset $D \subseteq Q$ if $D \cap (Q\cdot w) = \emptyset$.
Having a good upper bound on the length of the shortest words avoiding $D$, we could continue using avoiding words for subsets smaller than $n/2$, since for a word $s$ there are at least $|D|\cdot|Q\cdot s|-n$ states such that $1 \le |q\cdot s^{-1}| \le |D|$ (see Lemma~\ref{lem:preimage_states}).
\begin{openproblem}
Find a good upper bound (in terms of $|D|$ and $n$) on the length $\ell$ such that in every $n$-state automaton, for every subset $D \subset Q$ there is a word avoiding $D$ of length at most $\ell$, unless $D$ is not avoidable.
\end{openproblem}

In fact, we can prove an upper bound in the spirit of Lemma~\ref{lem:avoid_or_compress}, provided that we have avoiding words for smaller subsets than $D$.
\begin{lemma}\label{lem:avoiding_subset_or_compress}
For $n \ge 2$, let $\mathrsfs{A}(Q,\Sigma,\delta)$ be an $n$-state strongly connected synchronizing automaton.
Consider non-empty subsets $S,D \subseteq Q$ such that $|S| \ge 1$ and $|D| \ge 2$.
Suppose that there is a state $p \in D$ such that for $D'=D \setminus \{p\}$ there exists a word $w_{D'} \in \Sigma^{\ell}$ that avoids $D'$.
Then there exists a word $w \in \Sigma^{n-1+\ell}$ such that either:
\begin{enumerate}
\item $(S\cdot w) \cap D = \emptyset$, or
\item $|S\cdot w| < |S|$.
\end{enumerate}
\end{lemma}
\begin{proof}
Let $L_i = \lspan(\{[S][w] \mid w \in \Sigma^{\le i}\})$.
We consider the following sequence of linear subspaces:
$$L_0 \subseteq L_1 \subseteq L_2 \subseteq \ldots,$$
and use the ascending chain condition as in the proof of Lemma~\ref{lem:avoid_or_compress}.
Since the automaton is synchronizing, there is a reset word $u$ so $[S][u]=n[q]$ for some state $q$.
Since the automaton is strongly connected, for every state $p$ we have a word $v$ such that $q\cdot v = p$, and so $[S][uv]=n[p]$.
These vectors generate the whole space $\mathbb{R}^n$, and so the maximal dimension of the linear subspaces from the sequence is $n$; in particular, $\dim(L_{n-1})=n$.

Let $P = p\cdot (w_{D'})^{-1}$.
Suppose for a contradiction that for every word $w$ of length $\le n-1$, subset $S$ is not compressed by $w$ and $|(S\cdot w) \cap P|=1$.
Then $[S][w]$ contains exactly one $1$ and $|P|-1$ $0$s at the positions corresponding to the states from $P$.
Therefore, all vectors $v$ generated by the vectors with this property satisfy:
$$(|S|-1)\sum_{i \in P} v(i) = \sum_{i \in Q \setminus P} v(i).$$
This means that the dimension of $L_{n-1}$ is at most $n-1$, since in $\mathbb{R}^n$ there are vectors that broke this equality.
Hence, we have a contradiction.

Hence, there must be a word $w$ that either compresses $S$ or is such that $|(S\cdot w)\cap P| \ne 1$.
In the latter case, if $(S\cdot w)\cap P = \emptyset$ then we obtain $(S\cdot w w_{D'}) \cap D = \emptyset$.
If $(S\cdot w)\cap P \ge 2$ then $w_{D'}$ maps at least two states from $(S\cdot w)\cap P$ to $p$, thus $w w_{D'}$ compresses $S$.
\end{proof}

By an iterative application of the above lemma, we can obtain the upper bound $k(n-1+kn)$ on the length of a word that either avoids two states from the given subset or compresses the subset.
This bound is too large to provide a further improvement (at least within the cubic coefficient) for the upper bound on the length of the shortest reset words.
However, if the shortest words avoiding a single state are indeed of linear length, then we obtain a quadratic upper bound on the length of the shortest words avoiding two states.

\medskip\textbf{$\bullet$ Breaking a partition}:
From Corollary~\ref{cor:avoiding_bound} we see that the upper bound on the length of avoiding words is better when we have more states to avoid in the subset $A$. Also, the construction showing a lower bound $2n-3$ on the worst-case of avoiding contains only one state requiring such long avoiding words, whereas all the other states are avoidable with a word of length at most $n-1$.
Hence, it should be easier to obtain better upper bounds when we have more choices to avoid (larger subset $A$).

In fact, in our application we are never forced to avoid a single particular state (or subset), but rather we may choose from a number of possibilities.
If $w$ is a word of rank $r$, then it defines a partition of $Q$ into disjoint subsets $P_1,\ldots,P_r$ which are preimages of the states in $Q\cdot w$ under the action of $w$.
To obtain a word of rank $<r$, it is enough to find a word avoiding any of these $r$ subsets.
\begin{openproblem}
For a partition of $Q$ into disjoint subsets $P_1,\ldots,P_r$, what is the smallest length $\ell$ such that there exists a word of length at most $\ell$ avoiding at least one these subsets?
\end{openproblem}
An upper bound could be expressed e.g.\ in terms of $r$ or the maximum/minimum size of these subsets.

\medskip
\noindent\textbf{$\bullet$ Compressing a pair with a given state}:
Another related problem is to bound the length of a word compressing a given state with another state.
Given a state $q \in Q$, what is the length of the shortest words such that $q\cdot w = p\cdot w$ for some other state $p \neq q$; that is, $w$ compresses a pair of states containing $q$.
Note that for a given pair $\{p, q\}$, the shortest compressing words can have length up to $n(n-1)/2$ (the number of all pairs), which is the case in the \v{C}ern\'{y} automata \cite{Cerny1964}, but no such construction is known when $q$ must be compressed just with an arbitrary state $p \neq q$.
At the first glance it seems to be unrelated to avoiding words, but in fact, there is a dependency
between the bounds of the shortest compressing words and the lengths of the shortest avoiding words -- one can use
compressing words to construct an avoiding word and vice versa.
The main question here is whether there exists a linear upper bound on the length of the shortest compressing words (in particular, for strongly connected and synchronizing case).
A quadratic upper bound is obvious, and there are examples requiring linear length (e.g.\ the \v{C}ern\'{y} automata).
\begin{openproblem}
Find a good upper bound on the smallest length $\ell$ (in terms of $n$) such that in every synchronizing strongly connected $n$-state automaton, for every state $p$ there exists a state $q \neq p$ such that $\{p,q\}$ is compressible by a word of length at most $\ell$.
\end{openproblem}

\bigskip
\noindent\textbf{Acknowledgments}.
I thank Mikhail Berlinkov, Costanza Catalano, Vladimir Gusev, Jakub Ko{\'s}mider, and Jakub Kowalski for proofreading and comments.
\bibliographystyle{plainurl}

\end{document}